\pdfoutput=1

\documentclass[12pt,fleqn]{amsart}
\usepackage{amsmath,bm}
\usepackage{amsthm}

\usepackage{graphicx}
\usepackage{picins}

\usepackage{amsfonts}
\usepackage{latexsym}
\usepackage{amssymb}
\usepackage{graphicx}
\usepackage{picins}

\usepackage{url}		
\usepackage{xspace}		
\usepackage{booktabs}
\usepackage{subfigure} 
\usepackage{rotating}	
\usepackage{fancyvrb}

\newtheorem{definition}{Definition}
\newtheorem{thm}{Theorem}
\newtheorem{prop}{Proposition}
\newtheorem{coro}{Corollary}
\newtheorem{lemma}{Lemma}

\newtheoremstyle{break}
  {9pt}
  {9pt}
  {}
  {}
  {\bfseries}
  {}
  {\newline}
  {}

\graphicspath{{../..//}}

\newcommand{\reell}{{\mathbb R}}

\newcommand{\erw}{\textnormal{E}}

\newcommand{\NormVert}{\textbf{{N}}}

\newcommand{\FDR}{\textnormal{FDR}}
\newcommand{\FDP}{\textnormal{FDP}}
\newcommand{\FDPSU}{\textnormal{FDP-SU}}
\newcommand{\FDPSD}{\textnormal{FDP-SD}}

\newcommand{\FWER}{\textnormal{FWER}}
\newcommand{\kFWER}{\textnormal{k-FWER}}
\newcommand{\kFDP}{\textnormal{kFDP}}
\newcommand{\kFWERSU}{\textnormal{k-FWER-SU}}
\newcommand{\kFWERSD}{\textnormal{k-FWER-SD}}

\usepackage{natbib}

\begin{document}
\title[A sufficient criterion]{A sufficient criterion for control of generalised error rates in multiple testing}
\date{\today}
\begin{abstract}
Based on the work of \citet{RomanoShaikh2006,RomanoShaikh2006AOS,LehmannRomano2005} we give a sufficient criterion for controlling generalised error rates for arbitrarily dependent $p$-values. This criterion is formulated in terms of matrices associated with the corresponding error rates and thus it is possible to view the corresponding critical constants as solutions of sets of certain linear inequalities. This property can in some cases be used to improve the power of existing procedures by finding optimal solutions to an associated linear programming problem. 
\end{abstract}
\maketitle

\section{Introduction}
Consider the problem of testing $n$ hypotheses $H_1, \ldots, H_n$ simultaneously. A classical approach to dealing with the multiplicity problem is to control the familywise error rate (FWER) i.e. the probability of one or more false rejections. However, when the number $n$ of hypotheses is large, the ability to reject false hypotheses is small. Therefore, alternative type I error rates have been proposed that relax control of FWER in order to reject more false hypotheses (for a survey, see e.g. \citet{DudoitLaan2007}).

One such generalised error rate is the \kFWER, i.e. the probability of $k$ or more false rejections for some integer $k\ge 1$, $\kFWER\le \alpha$ considered by \citet{HommHoff87} and \citet{LehmannRomano2005}. For $k=1$ the usual FWER is obtained. Alternatively, instead of controlling the absolute number of false rejections it may be desirable to control the proportion of false rejections amongst all rejected hypotheses. This ratio is called the false discovery proportion (FDP). More specifically, if $R$ denotes the number of rejected hypotheses and $V$ the number of falsely rejected hypotheses, then $\FDP=V/R$ (and equal to $0$ if there are no rejections). For the $\FDP$, mainly two types of control have been considered in the literature. One aim might be to control the tail probability $P(\FDP>\gamma)\le \alpha$ for some user-specified value $\gamma \in [0,1)$. This error measure has been termed $\gamma-\FDP$ by \citet{LehmannRomano2005} and tail probability for the proportion of false positives ($TPPFP(\gamma)$) in \citet{DudoitLaan2007}. Instead of controlling a specific tail probability, the false discovery rate (FDR) requires that $\FDR=\erw(\FDP)\le \gamma$, i.e. control in the mean. As \citet{RomanoWolf2010} point out, probabilistic control of the FDP allows one to make useful statements about the realized FDP in applications, whereas this is not possible when controlling the FDR.
 
Recently, a number of methods have been proposed that control these generalised error rates under various assumptions. In this paper we focus on multiple testing procedures that are based on marginal $p$-values and are valid for finite sample sizes under no assumptions on the type of dependency of these $p$ values. For $\kFWER$ and $\gamma-\FDP$, step-up and step-down methods have been obtained in \citet{RomanoShaikh2006,RomanoShaikh2006AOS,LehmannRomano2005}. For $\FDR$, \citet{BenjaminiYekutieli01} have shown that a rescaled version of the original step-up procedure of \citet{BenjaminiHochberg95} controls the $\FDR$ under arbitrary dependencies. \citet{GuoRao2008} have extended these results and have also given corresponding upper bounds for step-down $\FDR$ procedures (see \citet{GuoRao2008} and the references cited therein for more details). 

The aim of this paper is two-fold. First, we present a sufficient condition for control of $\kFWER$ and $\gamma-\FDP$ based on matrices that are associated with a specific error-rate and direction of stepping. This result is mainly a rephrasing of results obtained by \citet{RomanoShaikh2006,RomanoShaikh2006AOS,LehmannRomano2005}. In the second step we show how the rescaled procedures introduced by \citet{RomanoShaikh2006,RomanoShaikh2006AOS,LehmannRomano2005} can in some cases be improved. In particular, we introduce a linear programming approach which uses the above-mentioned matrices.

The paper is organized as follows. First, we introduce some terminology and assumptions that will be used in what follows. In section three we will state the main theoretical results which will be used in the following section to define new modified FDP controlling procedures. Section \ref{sec:Proofs} contains the proof of the main theorem. In section \ref{sec:SimulationStudy} we investigate the power of the new modified procedures in a simulation setting and in section \ref{sec:EmpiricalApplications} we apply them to the analysis of empirical data. The paper concludes with a discussion. 

\section{Notation, Definitions and assumptions} \label{sec:NotationAssumptions}
In this section we introduce some terminology and assumptions that will be used in the sequel. When testing hypotheses $H_1, \ldots,H_n$, we assume that corresponding $p$-values $PV_1, \ldots, PV_n$ are available. For any true hypothesis $i$ we assume that the distribution of the $p$-values $PV_i$ is stochastically larger than a uniform rv, i.e.
\begin{align*}
P(PV_i\le u)&\le u
\end{align*}
for all $u \in (0,1)$. Let $PV_{(1)}\le \cdots \le PV_{(n)}$ denote the ordered $p$-values and $H_{(1)}, \ldots, H_{(n)}$ the associated (null-) hypotheses. Let 
\begin{align*}
\mathcal{C}&=\{c \in \reell_{+}^n | c_1 \le \cdots \le c_n\}
\end{align*}
denote the set of non-decreasing non-negative critical constants. For $c \in \mathcal{C}$ the associated step-up procedure rejects hypotheses $H_{(1)}, \ldots, H_{(k)}$, where $k= \max \{i| PV_{(i)} \le c_i\}$. If no such $i$ exists, no hypothesis is rejected. For the corresponding step-down procedure, reject $H_{(1)}, \ldots, H_{(k)}$, where $k= \max \{i| PV_{(j)} \le c_j, \quad j=1, \ldots, i\}$. If $PV_{(1)}>c_1$, no hypothesis is rejected. 

\subsection{Generalized error rates} \label{ssec:GeneralizedErrorRates}
In the following definition we introduce the sets of $\kFWER$- and $\FDP$-controlling procedures weconsider in this paper.
\begin{definition} \label{def:ControllingPorcedures}
Let $\alpha \in (0,1)$. 
\begin{itemize}
	\item[(a)] For $1 \le k \le n$ define the set of step-up procedures that (strongly) control the $\kFWER$:
	\begin{align*}
S^{\kFWERSU}(\alpha,k)&=	\{ c\in \mathcal{C} |\max_{1\le |I| \le n} \kFWERSU (c) \le \alpha\},\\
S^{\kFWERSD}(\alpha,k)&=	\{ c\in \mathcal{C} |\max_{1\le |I| \le n} \kFWERSD (c) \le \alpha\}
	\end{align*}
	\item[(b)] For $\gamma \in [0,1)$ define the set of step-up and step-down procedures that (strongly) control the $\FDP$:
	\begin{align*}
S^{\FDPSU}(\alpha,\gamma)&=	\{ c\in \mathcal{C} |\max_{1\le |I| \le n} P(\FDPSU(c)> \gamma) \le \alpha\},\\
S^{\FDPSD}(\alpha,\gamma)&=	\{ c\in \mathcal{C} |\max_{1\le |I| \le n} P(\FDPSD(c)> \gamma) \le \alpha\}.
	\end{align*}
\end{itemize}
\end{definition}
In order to formulate the main results of this paper we introduce subsets of $\mathcal{C}$ that are defined by 
 \begin{align*}
\mathcal{F}(A)&:= \{c \in \mathcal{C} | ||A\cdot c||_{\infty} \le 1\}, 
	\end{align*}
where $A \in \reell_{+}^{n \times n}$ and $||x||_{\infty}  =\max_{1\le i \le n} |x|_i$ denotes the maximum norm. The elements of $\mathcal{F}(A)$ can be interpreted as the set of feasible points given by a set of linear constraints (inequalities). We will show in Theorem \ref{theorem:MainTheorem} that for each error rate $\kFWER$ and $\FDP$ and direction of stepping we can define an associated matrix $A$, such that any procedure in $\alpha \cdot \mathcal{F}(A)$ controls the correponding error rate at level $\alpha$. 
\subsection{Associated matrices} \label{ssec:AssociatedMatrices}
In this section we introduce the matrices associated with the error rates mentioned above.
\begin{definition}[\kFWERSU] \label{def:MatrixkFWERSU}
Let $k \in \{1,\ldots,n\}$. Define
\begin{align}
A^{\kFWERSU}_{ij}(k)&=i\cdot
\begin{cases}
0 & \qquad i<k\\
0 & \qquad i\ge k, j< n+k-i\\
 \left(\frac{1}{j-n+i} -\frac{1}{j-n+i+1}\right) & \qquad i\ge k, n+k-i \le j <n\\
\frac{1}{i} & \qquad i\ge k,  j =n
\end{cases}
\end{align}
\end{definition}
\begin{definition}[\kFWERSD] \label{def:MatrixkFWERSD}
Let $k \in \{1,\ldots,n\}$. Define
\begin{align}
A^{\kFWERSD}_{ij}(k)&= i \cdot
\begin{cases}
\frac{1}{k} & \qquad i\ge k, j=n-i+k\\
0 & \qquad \text{else}
\end{cases}
\end{align}
\end{definition}

\begin{definition} \label{def:AuxFDPSU}
Let $\gamma \in [0,1)$ and define $m(j)=\lfloor \gamma j \rfloor +1$ where $\lfloor x \rfloor$ is the greatest integer $\le x$. For $i \in \{1, \ldots,n\}$ define 
\begin{align*}
\widetilde{M}(i)&:=\widetilde{M}(\gamma,n,i)=\max\{l \in \{1,\ldots,n\}| m(l) \le i\}\\
g_i(l)&:= \max \{i-n+l,m(l)\}\\
M(i)&:=g_i(\widetilde{M}(i))\\
t_k (i)&:=\max g_i^{-1}(\{k\}), \text{ for } k \in \{1,\ldots,M(i)\}.
\end{align*} 
\end{definition}

\begin{definition}[\FDPSU] \label{def:MatrixFDPSU}
Let $\gamma \in [0,1)$. Using the notation from definition \ref{def:AuxFDPSU}, define
\begin{align}
A^{\FDPSU}_{ij}(\gamma) &=i \cdot 
\begin{cases}
\left(\frac{1}{k}- \frac{1}{k+1} \right) & \text{for $j=t_k(i)$, if $1\le k < M(i)$}\\
\frac{1}{M(i)} & \text{for $j=t_{M(i)}(i)$,}\\
0 & \text{else.}
\end{cases}
\end{align}
\end{definition}

\begin{definition} \label{def:AuxFDPSD}
Let $c \in\mathcal{C}$. For $i \in \{1, \ldots,n\}$ define 
\begin{align*} 
k_{i}(l)&= \min \{n,n+l-i,\left\lceil\frac{l}{\gamma}\right\rceil-1\}\qquad l=1, \ldots,\lfloor \gamma n \rfloor+1 \qquad \text{and}\\
N(i)&=N(\gamma,n,i)=\min \{ \lfloor \gamma n \rfloor +1,i, \left\lfloor \gamma \cdot \left( \frac{n-i}{1-\gamma}+1\right) \right\rfloor +1 \} ,\\
M(i)&=k_i(\{1, \ldots,\lfloor \gamma n \rfloor+1\}).
\end{align*}
\end{definition}

\begin{definition}[\FDPSD] \label{def:MatrixFDPSD}
Let $\gamma \in [0,1)$. Using the notation from definition \ref{def:AuxFDPSD} define 
\begin{itemize}
\item[(a)] $\widetilde{A} \in \reell^{n \times n}$ by
\begin{align*}
\widetilde{A}_{ij} &=i \cdot
\begin{cases}
 \left(\frac{1}{j}- \frac{1}{j+1}\right) & \qquad \text{for $1\le j < N(i)$},\\
\frac{1}{N(i)} & \qquad \text{for $j=N(i)$},\\
0 & \qquad \text{else}.
\end{cases}
\end{align*}
	\item[(b)] and
\begin{align*}
A^{\FDPSD}_{ij}(\gamma)&=
\begin{cases}
\sum_{l \in k_i^{-1}(\{j\})}\widetilde{A}_{il}& \qquad j \in M(i),\\
0 & \qquad \text{else.}
\end{cases}
\end{align*}
\end{itemize}
\end{definition}

%
%
\section{Main results} \label{sec:TheoreticalResults}
First we state the main results of this paper which will serve as the starting point for mofifying some existing MTPs. As the proof in section \ref{sec:Proofs} shows, it is actually a rephrasing of results of \citet{RomanoShaikh2006,RomanoShaikh2006AOS} in terms of the associated matrices introduced in section \ref{sec:NotationAssumptions}.
\begin{thm}\label{theorem:MainTheorem}
Let $\alpha \in (0,1)$.
\begin{itemize}
\item[(a)] For $1\le k \le n$ it holds $\alpha \cdot\mathcal{F}(A^{\kFWERSU}(k)) \subset S^{\kFWERSU}(\alpha,k)$.
\item[(b)] For $1\le k \le n$ it holds $\alpha \cdot\mathcal{F}(A^{\kFWERSD}(k)) \subset S^{\kFWERSD}(\alpha,k)$.
\item[(c)] For $\gamma \in [0,1)$ it holds $\alpha \cdot\mathcal{F}(A^{\FDPSU}(\gamma)) \subset S^{\FDPSU}(\alpha,\gamma)$.
\item[(d)] For $\gamma \in [0,1)$ it holds $\alpha \cdot\mathcal{F}(A^{\FDPSD}(\gamma)) \subset S^{\FDPSD}(\alpha,\gamma)$.
\end{itemize}
\end{thm}
The theorem provides generic sufficient conditions for control of generalised error rates, i.e. if $d \in \mathcal{F}(A)$ then $\alpha \cdot d$ controls the corresponding error rate at the desired level. Since the sets $\mathcal{F}(A)$ from the theorem are convex, it follows immediately that for any matrix $A$ from theorem \ref{theorem:MainTheorem} and level $\alpha \in (0,1)$ the set of procedures $\alpha \cdot \mathcal{F}(A)$ is also convex. \citet{GuoHeSarkar2012} have introduced the $\gamma-\kFDP=P(\kFDP> \gamma)$ where $\kFDP=V/R$ ($V$ and $R$ defined as in the introduction) if $V \ge k$ and $0$ else. Under the assumption that $PV_i \sim U(0,1)$ under any true hypothesis $i$ they obtain linear bounds for the $\gamma-\kFDP$ in the proofs of their Theorems 4.1 and 4.2.  These bounds can again be used to define appropriate associated matrices and establish a result similar to the above theorem for the $\gamma-\kFDP$ but we do not pursue this any further here. 

One immediate consequence of the theorem is the following corollary.
\begin{coro} \label{coro:RescalingProcedures}
Let $1\le k \le n$, $\gamma \in [0,1)$. For $c \in \mathcal{C}$ and $x \in \{\kFWERSU, \kFWERSD, \FDPSU,\FDPSD\}$ define 
\begin{align*}
D^x(c) &= ||A^x \cdot c||_{\infty}.
\end{align*}
Then the rescaled procedure $\widetilde{c}:= \alpha \cdot c / D^x(c)$ yields control of the error rate $x$ at level $\alpha$.
\end{coro}
Thus we can always achieve control of generalised error rates by using the rescaling approach.

The proof of the above theorem relies on two key tools. The first is the following generalised Bonferroni inequality due to \citet{LehmannRomano2005}.
\begin{lemma}[\citet{LehmannRomano2005}] \label{lemma:BasicLemma}
Let $X_1,\ldots,X_t: \Omega \rightarrow (0,1]$ be $p$-values that satisfy the above distributional assumption i.e. $P(X_i\le u)\le u$ for all $i$ and $u \in (0,1)$. Denote their ordered values by $X_{(1)}\le \cdots \le X_{(t)}$ and let $0=c_0 \le c_1 \le \cdots \le c_m\le 1$ for some $m \le t$. 
\begin{itemize}
	\item[(i)] Then it holds
\begin{align}
P(\{X_{(1)}\le c_1 \} \cup \cdots \cup  \{X_{(m)}\le c_m \} ) &\le t \cdot\sum_{i=1}^{m} \frac{c_i - c_{i-1}}{i}. \label{eq:GenBonf}
\end{align}
	\item[(ii)] As long as the right-hand side of \eqref{eq:GenBonf} is $\le 1$, the bound is sharp in the sense that there exists a joint distribution for the $p$-values for which the inequality is an equality.
\end{itemize}
\end{lemma}
Related inequalities have been obtained previously by \citet{Hommel1983}, \citet{RoehmelStreitberg1987} and \citet{Falk1989}.

The second step uses the observation that the generalised error rates considered here can all be bounded by probabilities of the type 
\begin{align}
P &\left( \bigcup_{i =1}^{M(|I|)} \{PV_{(i)} \le c_{t_i(|I|)}\} \right), \label{eq:GenProbBound}
\end{align}
where $|I|$ is the number of true hypotheses, $M(|I|)\in \{0,\ldots,n\}$ and $t_i(|I|) \in \{0,\ldots,n\}$ is an increasing sequence in $i$ (depending on $|I|$) and the $PV$ in \eqref{eq:GenProbBound} are taken under the null hypotheses. Then the probability in \eqref{eq:GenProbBound} can be bounded using lemma \ref{lemma:BasicLemma}. We call the resulting bound the LR-bound of the corresponding error rate.

For the procedures considered here, adjusted $p$-values can be defined in the generic way decribed in \citep[Procedures 1.3 and 1.4]{DudoitLaan2007}: For raw $p$-values $pv_1, \ldots, pv_n$ and $c=(c_1, \ldots,c_n) \in \mathcal{F}(A)$ define step-up $p$-values
\begin{align}
\widetilde{pv}_{(i)} &= \min_{j=i, \ldots,n} \left\{ \min \left( \frac{pv_{(j)}}{c_j},1\right)  \right\}
\intertext{and step-down $p$-values}
\widetilde{pv}_{(i)} &= \max_{j=1, \ldots,i} \left\{ \min \left( \frac{pv_{(j)}}{c_j},1 \right)  \right\} .
\end{align}

In what follows we will focus on $\FDP$ controlling procedures.

\section{Modified FDP-controlling procedures} \label{sec:NewFDPProcedures}
In addition to providing an easily verifiable condition for FDP controlling procedures, theorem \ref{theorem:MainTheorem} can be used to construct new or modify existing procedures. In this section we describe an approach based on linear programming. Our focus in this section is on improving classical procedures based on rescaled constants as considered in \citet{RomanoShaikh2006,RomanoShaikh2006AOS}. First we define new modified FDP procedures as the solutions of a linear programming problem.
\begin{definition}\label{def:NewGenericFDPProcedure}
Let $A \in \reell_{+}^{n \times n}$ and $c \in \mathcal{F}(A)$. Define the modified procedure $\xi=\xi(c) $ as the solution to the following linear programming problem (P):
\begin{align}
\text{maximize} \qquad & F(\xi)= a \cdot \xi   &&\notag\\
\text{subject to} \qquad & A_{i \cdot} \cdot \xi \le 1 &&\qquad i=1, \ldots, n \tag{P}\\
& -\xi_i+\xi_{i-1} \le 0  &&\qquad i=1, \ldots, n \notag\\ 
& -\xi_i+c_{i} \le 0  &&\qquad i=1, \ldots, n, \notag
\end{align}
where $\xi_0=0$ and $a_j=\sum_{i=1}^n A_{ij}$.
\end{definition}
Note that the third constraint in (P) implies that $\xi\ge c$ while the first and second constraints guarantee that $\xi \in \mathcal{F}(A)$. Note also that if $c=0$ then $\mathcal{F}(A)$ is identical with the feasible points of the optimisation problem, so that this approach could be used to find optimal solutions within the whole class $\mathcal{F}(A)$ instead of $\mathcal{F}(A) \cap \{\xi \ge c\}$. Since we are primarily interested in improving existing procedures we do not pursue this any further. 

For problems like (P), standard numerical methods like the simplex algorithm \citep{Dantzig63} are available. From a statistical viewpoint, it would be desirable to optimise the power of the MTP (defined in a suitable sense, see also section \ref{sec:SimulationStudy}), subject to the given constraints. The rationale for using the objective function $F$ is the following: Let $b_i=\sum_{j=1}^n A_{ij} \cdot \xi_j$, so that by Theorem \ref{theorem:MainTheorem} under $|I|=i$ the error rate is bounded by $b_i$ and the sum $b_1+ \cdots +b_n=F(\xi)$ can thus be interpreted as the sum of the maximum significance levels of the procedure. Since we are aiming for a powerful procedure it seems plausible to optimise this objective function in the sense that the best we can do without violating the bounds from lemma \ref{lemma:BasicLemma} is $F(\xi)=n$. Thus $F(\xi)$ may be thought of as a surrogate-measure of power. It can also be interpreted in a Bayesian framework by observing that optimising it is equivalent to optimising the mean maximum level of significance if the number of true hypotheses $|I|$ is distributed uniformly on $\{1, \ldots,n\}$. Thus, if prior knowledge is available for the distribution of $|I|$, we could also use the weighted objective function $F(\xi,w)=\sum_{i=1}^n w_i \cdot b_i$ where $w_i=P(|I|=i)$.

Using theorem \ref{theorem:MainTheorem} we immediately obtain the following result.
\begin{coro}
Let $\gamma \in [0,1)$, $A \in \{ A^{\FDPSU}(\gamma),A^{\FDPSD}(\gamma)  \}$ and $c \in \mathcal{F}(A)$. Let $\xi=\xi(c)$ as defined in definition \ref{def:NewGenericFDPProcedure}. Then $\xi \in \mathcal{F}(A)$ and therefore the procedure $\alpha \cdot \xi$ controls the FDP for any $\alpha\in (0,1)$. This procedure is at least as powerful as procedure $\alpha\cdot c$.  
\end{coro}
Clearly, if $F(\xi) > F(c)$, then $\xi>c$. This means that this approach will always find a strict improvement over $c$ whenever one exists and we may thus expect a gain in power. Since, by construction, $\xi$ can not be improved uniformly within class $\mathcal{F}(A)$, $\alpha\cdot \xi$ can be seen as an optimal procedure within the subset $\alpha \cdot \mathcal{F}(A)$ of all $\alpha$-controlling FDP procedures. 

We now consider two specific types of critical constants in more detail.
\begin{itemize}
	\item[(a)] The Benjamini-Hochberg constants:
	\begin{align*}
	c_i^{BH}&=c_i^{BH}(n)=\frac{i}{n} 
	\end{align*}
	\item[(b)] The Lehmann-Romano constants:
	\begin{align*}
	c_i^{LR}&=c_i^{LR}(\gamma,n)=\frac{\lfloor \gamma i \rfloor +1}{n + \lfloor \gamma i \rfloor +1 -i} 
	\end{align*}
\end{itemize}
In \citet{RomanoShaikh2006,RomanoShaikh2006AOS} normalising constants were introduced for $c^{BH}$ and $c^{RS}$ for step-up and step-down procedures. These constants were defined (in our notation) by
\begin{align*}
D^{BH-SU}(\gamma) &= ||A^{\FDPSU}(\gamma) \cdot c^{BH}||_{\infty},\\
D^{RS-SU}(\gamma) &= ||A^{\FDPSU}(\gamma) \cdot c^{RS}(\gamma)||_{\infty},\\
D^{BH-SD}(\gamma) &= ||A^{\FDPSD}(\gamma) \cdot c^{BH}||_{\infty},\\
D^{RS-SD}(\gamma) &= ||A^{\FDPSD}(\gamma) \cdot c^{RS}(\gamma)||_{\infty},
\end{align*}
and due to corollary \ref{coro:RescalingProcedures} the rescaled procedures $\alpha \cdot c/D(\gamma)$ all control the $\gamma$-FDP at level $\alpha$.  
\subsection*{Example} Figure \ref{fig:Explanationl150} illustrates the possible gains resulting from the optimisation approach for $n=50$ and $\gamma=0.05$.
\begin{sidewaysfigure}[htbp]
	\centering
		\includegraphics[width=1.00\textwidth]{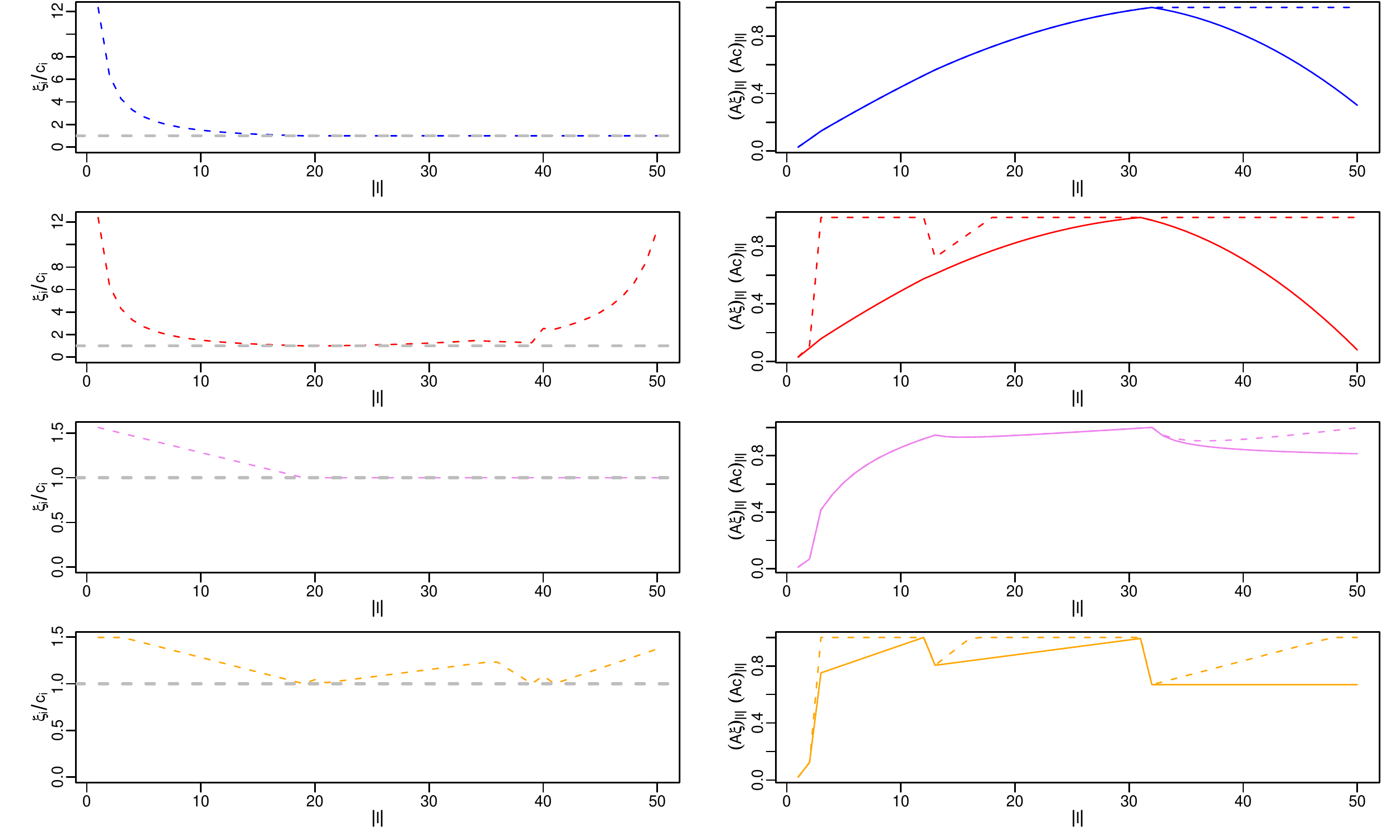}
	\caption{Illustration for $\FDP$ with $n=50$ and $\gamma=0.05$. Left panels: Ratios of modified critical constants $\xi$ to original rescaled constants $c$. Right panels: Values of $(A \cdot \xi)_{|I|}$ (dashed lines) and $(A \cdot c)_{|I|}$ (solid lines). The procedures BH-SU, BH-SD, RS-SU and RS-SD are sorted from top to bottom.}
	\label{fig:Explanationl150}
\end{sidewaysfigure} 
In all cases the modified procedures are strictly better than the rescaled procedures. To investigate where the gains come from, we consider the BH-SU procedure in more detail. For the rescaled procedure $c=c^{BH}/D^{BH-SU}(\gamma)$, $(A^{\FDPSU}(0.05)\cdot c)_{|I|}=1$ for $|I|=32$. The column entries for row 32 of matrix $A^{\FDPSU}(0.05)$ are strictly greater than zero for columns $19$ to $50$ and therefore the associated critical constants $c_{19}, \ldots,c_{50}$ can not be improved upon (any increase would violate the constraint $\max (A \cdot c)_{|I|}\le 1$). However, since $A_{32,1}=\cdots =A_{32,18}=0$ there is some potential for increasing the remaining critical constants $c_{1}, \ldots,c_{18}$. This is exactly what the optimisation in the linear program (P) accomplishes. Ideally, this would result in a new procedure $\xi$ with $(A \cdot \xi)_{|I|}= 1$ for all $|I|$, yielding a completely unimprovable procedure within class $\mathcal{F}(A)$. This happens e.g. for $A=A^{\kFWERSD}$, when $\xi$ is the vector of Lehmann-Romano constants, see section \ref{ssec:ProofMainTheoremb}. However, due to the structure of the matrix $A$, this is usually impossible. In the case of BH-SD we obtain $A\cdot \xi_{32}= \cdots =A\cdot \xi_{50}=1$ (see uppermost right panel in figure \ref{fig:Explanationl150}).

Figure \ref{fig:Explanationl150} suggests that the gains derived from the modifications are considerably larger for the BH than for the RS procedures. This is also supported by the numerical values in table \ref{tab:CompareProcedures}. If we follow the arguments given above for justifying the choice of objective function we would expect the modified BH-SD procedure to be the most powerful procedure (indicated by the highest values of $F(\xi)$), followed closely by the modified RS-SD procedure. This is also consistent with the simulation results in section \ref{sec:SimulationStudy}. 

\begin{sidewaystable}[htb]
	\centering
\begin{tabular}{rrrrrrrrrrrrrrrrr}
\toprule
           &                                                                               \multicolumn{ 8}{c}{SU} &                                                                               \multicolumn{ 8}{c}{SD} \\ \cmidrule(lr){2-9}\cmidrule(lr){10-17} 

           &                           \multicolumn{ 4}{c}{BH} &                           \multicolumn{ 4}{c}{RS} &                           \multicolumn{ 4}{c}{BH} &                           \multicolumn{ 4}{c}{RS} \\ \cmidrule(lr){2-5}\cmidrule(lr){6-9}\cmidrule(lr){10-13}\cmidrule(lr){14-17}
         $n$ &        $F(c)$ &       $F(\xi)$ &   $M_1$ & $M_2$ &         $F(c)$ &       $F(\xi)$ &   $M_1$ & $M_2$ &      $F(c)$ &       $F(\xi)$ &   $M_1$ & $M_2$ &    $F(c)$ &       $F(\xi)$ &   $M_1$ & $M_2$   \\ \hline
        10 &       7.75 &       8.16 &       2.61 &       1.34 &      8.76 &       8.76 &       1.00 &       1.00 &       7.33 &       10.00 &       3.00 &       3.00 &      10.00 &       10.00 &       1.00 &       1.00 \\

        25 &       18.32 &       20.39 &     6.42 &       2.09 &       21.32 &       22.75 &       1.23 &       1.11 &       17.18 &      24.14 &       6.76 &       6.76 &       17.90 &       23.50 &       1.43 &       1.43 \\

        50 &       32.78 &       37.90 &    12.36 &       3.14 &       41.75 &       43.39 &       1.56 &       1.23 &       31.55 &       48.17 &      12.4 &       12.4 &       38.69 &       44.94 &       1.50 &       1.50 \\

       100 &       66.97 &       74.02 &    18.39 &       3.92 &       83.63 &       85.47 &       1.29 &       1.09 &       65.24 &      94.89 &      18.39 &       18.39 &       77.47 &      87.01 &       1.57 &       1.51 \\

       250 &       165.51 &       173.72 &  19.00 &       3.59 &       207.72 &       209.11 &       1.11 &       1.03 &       164.27 &      230.50 &      24.41 &       19.00 &      196.77 &      219.11 &       1.82 &       1.71 \\

       500 &       328.09 &       336.90 &  19.00 &       3.34 &      411.57 &      412.68 &       1.05 &       1.02 &       328.13 &      459.61 &      31.05 &       19.00 &      392.67 &      444.89 &       2.15 &       1.96 \\

      1000 &    650.00    &       659.18 &  19.00  &    3.13        &   812.64   &    813.49       &  1.03          &    1.01        &    653.11       &   921.70        &   39.52         &    19.00        &   778.33         &   902.52         &    2.49        &      2.22      \\
\bottomrule
\end{tabular}  
\caption{Values of $F$ for rescaled and modified critical constants and maximum ratios $M_1=\max \xi_i/c_i$, $M_2=\max (A \cdot \xi)_{|I|}/(A \cdot c)_{|I|}$ for $\FDP$ procedures with $\gamma=0.05$.}
\label{tab:CompareProcedures}
\end{sidewaystable}
\clearpage

\section{Proofs}\label{sec:Proofs}
In this section we prove the statements of the theorem. Actually, the main work is to rephrase the results of \citet{RomanoShaikh2006,RomanoShaikh2006AOS} in terms of the matrices introduced in section \ref{sec:NotationAssumptions}. The structure of the proofs is the same in all cases.
\subsection{Proof of theorem \ref{theorem:MainTheorem}, part (a)}
\begin{proof}
Let $d \in \mathcal{F}(A^{\kFWERSU}(k))$, define $c=\alpha \cdot d$ and let $I \subset \{1,\ldots,n\}$ be the set of true hypotheses. By \citet[lemma 3.1]{RomanoShaikh2006AOS} we have 
\begin{align*}
\kFWER(c) &\le P\left( \bigcup_{k\le \ell\le |I|} \{PV_{(\ell)} \le c_{n-|I|+\ell}\} \right)
\end{align*}
where $PV_{(1)},\ldots,PV_{(|I|)}$ are the $p$-values corresponding to the null hypotheses. By lemma \ref{lemma:BasicLemma} with $t=|I|=m$ and $\tilde{c}_0=\ldots=\tilde{c}_{k-1}=0,\tilde{c}_k=c_{n-|I|+k},\ldots,\tilde{c}_{|I|}=c_n$ this probability can be bounded by
\begin{align}
&|I| \cdot \left\{\frac{c_{n-|I|+k}}{k} + \frac{c_{n-|I|+k+1} -c_{n-|I|+k}}{k+1} + \cdots + \frac{c_{n} -c_{n-1}}{|I|} \right\} \notag\\
&= |I| \cdot \left\{c_{n-|I|+k}\cdot \left(\frac{1}{k}-\frac{1}{k+1} \right) + \cdots + c_{n-1}\cdot \left(\frac{1}{|I|-1}-\frac{1}{|I|}\right) + c_n \cdot \frac{1}{|I|} \right\} \label{eq:kFWERMatrix1}\\ 
&= \sum_{j=1}^n A_{|I|j}\cdot c_j, \label{eq:kFWERMatrix2}\\
&= \alpha \cdot (A\cdot d)_{|I|}\notag
\end{align}
where $A=A^{\kFWERSU}(k)$. Equality \eqref{eq:kFWERMatrix2} can be verified by considering the four cases in definition \ref{def:MatrixkFWERSU} separately:
\begin{itemize}
	\item For $|I|<k$ definition \ref{def:MatrixkFWERSU} yields $A_{|I|1}=\cdots=A_{|I|n}=0$ so that the bound in \eqref{eq:kFWERMatrix2} is equal to 0 which is correct, since $\kFWER(c)=0$ for $|I|<k$.
	\item For $|I|\ge k$ and $j<n+k-|I|$ the coefficient of $c_j $ is easily seen to equal 0 from equation \eqref{eq:kFWERMatrix1}.
	\item For $|I|\ge k$ note that the sum \eqref{eq:kFWERMatrix2} can be reexpressed as $\sum_{p=0}^{|I|-k} c_{n-|I|+k+p} \cdot A_{|I|,n-|I|+k+p}$. For $n+k-|I|\le j <n$ ($\Leftrightarrow$ $0\le p <|I|-k$) the coefficient of $c_{n-|I|+k+p}$ is $A_{|I|,n-|I|+k+p}=|I| \cdot (1/(k+p)-1/(k+p+1)$. Since $k+p=j+|I|-n$ the claim follows from the third part of the definition of $A$.
	\item For $|I|\ge k$ and $l=n$ the coefficient equals 1 as seen from equation \eqref{eq:kFWERMatrix1}.
\end{itemize}
Since $d \in \mathcal{F}(A^{\kFWERSU}(k))$, it follows
\begin{align*}
\max_{I \subset \{1,\ldots,n\}} \kFWER(c) &\le \alpha \cdot \max_{I \subset \{1,\ldots,n\}} (A \cdot d)_{|I|} \\
&\le \alpha
\end{align*}
\end{proof}
\subsection{Proof of theorem \ref{theorem:MainTheorem}, part (b)} \label{ssec:ProofMainTheoremb}
\begin{proof}
To prove that $\alpha \cdot \mathcal{F}(A^{\kFWERSD}(k)) \subset S^{\kFWERSD}(\alpha,k)$ let $d \in \mathcal{F}(A^{\kFWERSD}(k))$, define $c=\alpha \cdot d$ and let $I \subset \{1,\ldots,n\}$ be the set of true hypotheses. From the proof of Theorem 2.2 in \citet{LehmannRomano2005} it follows that
\begin{align*}
\kFWER(c) &\le P(PV_{(k)} \le c_{n-|I|+k})
\end{align*}
and by lemma \ref{lemma:BasicLemma} with $t=|I|$, $m=k$ and $0=\tilde{c}_0=\cdots=\tilde{c}_{m-1},\tilde{c}_m=c_{n-|I|+k}$ this probability can be bounded by
\begin{align}
\frac{|I|}{k}\cdot c_{n-|I|+k} &= \alpha \cdot (A\cdot d)_{|I|}\notag
\end{align}
where $A=A^{\kFWERSD}(k)$. 

To prove that $ S^{\kFWERSD}(\alpha,k)  \subset\alpha \cdot \mathcal{F}(A^{\kFWERSD}(k)$ we use the optimality property of the Lehmann-Romano procedure. Let $c \in S^{\kFWERSD}(\alpha,k)$. By Theorem 2.3 (ii) in \citet{LehmannRomano2005} it follows that for $i\ge k$ $c_i \le \alpha \cdot d^{LR}_i$ where $d^{LR}_i=k/(n+k-i)$ are the Lehmann-Romano critical constants. Now let $|I|\in \{1, \ldots,n\}$. Then it follows
\begin{align*}
(A \cdot d^{LR})_{|I|} &= A_{|I|,n-|I|+k} \cdot d^{LR}_{n-|I|+k}= \frac{|I|}{k} \cdot \frac{k}{n+k-(n-|I|+k)}=1
\end{align*}
so that $d^{LR} \in \mathcal{F}(A^{\kFWERSD}(k))$ and the claim is proved.
\end{proof}
\subsection{Proof of theorem \ref{theorem:MainTheorem}, part (c)}The following lemma is a re-phrasing of Lemma 4.1 in \citet{RomanoShaikh2006AOS} and states that the event $\{\FDR> \gamma\}$ is a subset of the union of sets of the type $\{PV_{(j)} \le c_{i_j}\}$. 
\begin{lemma}\label{lemma:FDPRepresentation}
Let the notation from definition \ref{def:AuxFDPSU} be given. Consider testing $n$ null hypotheses, with $|I|\ge 1$ of them true. Let $PV_{(1)},\ldots,PV_{(|I|)}$ denote the sorted $p$-values under the null hypotheses and let $\gamma \in [0,1)$. Then it holds for the step-up procedure based on the constants $c_1 \le \cdots \le c_n\le 1$
\begin{align*}
\{\FDP(c) > \gamma\} &\subset  \bigcup_{k =1}^{M(|I|)} \{PV_{(k)} \le c_{t_k (|I|)}\}
\end{align*} 
\end{lemma}
\begin{proof}
We use the bound 
\begin{align*}
\{\FDP > \gamma\} &\subset  \bigcup_{0 \le j \le n-1, |I|\ge m(n-j)} \{PV_{(\max[(|I|-j),m(n-j)])} \le c_{n-j}\}
\end{align*} 
given at the bottom of p. 1861 in \citet{RomanoShaikh2006AOS}. With $\ell=n-j$ the index set is now $\ell=1, \ldots,n$ with $m(\ell)\le |I|$ and so we have 
\begin{align*}
\{\FDP > \gamma\} &\subset  \bigcup_{\ell=1}^{\widetilde{M}(|I|)} \{PV_{(max[(|I|-n+\ell),m(\ell)])} \le c_{\ell}\}\\
&=\bigcup_{\ell=1}^{\widetilde{M}(|I|)} \{PV_{(g_{|I|}(\ell))} \le c_{\ell}\}
\end{align*} 
where the last equality follows from the definition of $g_{|I|}$ (see definition \ref{def:AuxFDPSU}), defined on $\{1, \ldots, \widetilde{M}(|I|)\}$. Clearly, $g_{|I|}$ is non-decreasing. Since $g_{|I|}(\ell+1)-g_{|I|}(\ell) \le 1$ and $g_{|I|}(1)=1$ it follows that $g_{|I|}(\ell) \le \ell$ and from the definition of $M_{|I|}$ that $g_{|I|}(\{1, \ldots, \widetilde{M}(|I|)\})=\{1, \ldots, M(|I|)\}$. 

For $k \in \{1, \ldots, M(|I|)\}$ we now claim that $\{PV_{(g_{|I|}(\ell))} \le c_{\ell}\} \subset \{PV_{(k)} \le c_{t_k (|I|)} \}$ for any $\ell \in g_{|I|}^{-1}(\{k\})$. To see this, let $\ell \in g_{|I|}^{-1}(\{k\})$. By the definition of $t_k$ it follows $\ell \le t_k(|I|)$. We thus obtain 
\begin{align*}
\{PV_{(g_{|I|}(\ell))} \le c_{\ell} \} &= \{PV_{(k)} \le c_{\ell}\} \qquad \text{(since $g_{|I|}(\ell)=k$)}\\
&\subset \{ PV_{(k)} \le c_{t_k (|I|)} \},
\end{align*}
since $\ell \le t_k (|I|)$. Altogether this yields
\begin{align*}
\{\FDP > \gamma\} &\subset  \bigcup_{\ell=1}^{\widetilde{M}(|I|)} \{PV_{(g_{|I|}(\ell))} \le c_{\ell}\}\\
&\subset \bigcup_{k=1}^{M(|I|)}\{PV_{(k)} \le c_{t_k (|I|)}\}.
\end{align*} 
\end{proof}
\begin{proof}[Proof of theorem \ref{theorem:MainTheorem}, part (c)] 
Let $d \in \mathcal{F}(A^{\FDPSU}(\gamma))$, define $c=\alpha \cdot d$ and let $I \subset \{1,\ldots,n\}$ be the set of true hypotheses. Be lemma \ref{lemma:FDPRepresentation} we have
\begin{align*}
\{\FDP(c) > \gamma\} &\subset  \bigcup_{k =1}^{M(|I|)} \{PV_{(k)} \le c_{t_k (|I|)}\}
\end{align*} 
and by lemma \ref{lemma:BasicLemma} this probability can be bounded by
\begin{align}
&|I| \cdot \left\{\frac{c_{t_1 (|I|)}}{1} + \frac{c_{t_2 (|I|)} -c_{t_1 (|I|)}}{2} + \cdots + \frac{c_{t_{M(|I|)} (|I|)}-c_{t_{M(|I|)-1} (|I|)}}{M(|I|)} \right\} \notag\\
&= |I| \cdot \left\{ c_{t_1 (|I|)}\cdot \left(1-\frac{1}{2} \right) + \cdots + c_{t_{M(|I|)-1}(|I|)}\cdot \left(\frac{1}{M(|I|)-1}-\frac{1}{M(|I|)}\right) + \frac{c_{t_{M(|I|)} (|I|)}}{M(|I|)} \right\} \label{eq:FDPSUMatrix1}\\ 
&= \sum_{j=1}^n A_{|I|j}\cdot c_j, \label{eq:FDPSUMatrix2}\\
&= \alpha \cdot (A\cdot d)_{|I|}\notag
\end{align}
where $A=A^{\FDPSU}(\gamma)$. Equality \eqref{eq:FDPSUMatrix2} can be verified by considering the following two cases:
\begin{itemize}
	\item If $M(|I|)=1$, the above upper bound equals $|I|\cdot c_{t_1 (|I|)}$ which is identical with \eqref{eq:FDPSUMatrix2} due to the second case in definition \ref{def:MatrixFDPSU}.
	\item If $M(|I|)>1$, the sum in \eqref{eq:FDPSUMatrix2} consists only of terms with $j=t_1 (|I|),t_2 (|I|), \ldots, t_{M(|I|)} (|I|)$ and the non-zero entries of row $|I|$ of $A$ are exactly the coefficents of $c_{t_1 (|I|)}, \ldots , c_{t_{M(|I|)} (|I|)}$ in \eqref{eq:FDPSUMatrix1}, corresponding to the first case in definition \ref{def:MatrixFDPSU}.
	\end{itemize}
Since $d \in \mathcal{F}(A^{\FDPSU}(\gamma))$ it now follows
\begin{align*}
\max_{I \subset \{1,\ldots,n\}} P(\{\FDP(c) > \gamma\}) &\le \alpha \cdot \max_{I \subset \{1,\ldots,n\}} (A \cdot d)_{|I|} \\
&\le \alpha
\end{align*}
\end{proof}

\subsection{Proof of theorem \ref{theorem:MainTheorem}, statement (d)}
The following is a rephrasing of results from \citet{RomanoShaikh2006}.
\begin{prop} \label{prop:FDP.SD.bound}
Let the notation from definition \ref{def:AuxFDPSD} be given and let $c \in\mathcal{C}$. For $1\le |I|\le n$ define
\begin{align*} 
\beta_\ell &= \beta_\ell(|I|)=c_{k_{|I|}(l)}, \qquad \ell=1, \ldots,\lfloor \gamma n \rfloor+1.
\intertext{Then it holds}
P(\FDP(c)>\gamma) &\le |I| \cdot \sum_{i=1}^{N(|I|)} \frac{\beta_i(|I|)-\beta_{i-1}(|I|)}{i}.
\end{align*}
\end{prop}
\begin{proof} Note that $N(i)$ from definition \ref{def:AuxFDPSD} is identical to (3.11) in \citet{RomanoShaikh2006}, $k_i$ corresponds to $k(s,\gamma,m,|I|)$ on p. 42 there, and $\beta$ defined above agrees with $\beta$ in (3.15) in \citet{RomanoShaikh2006}. As noted by \citet{RomanoShaikh2006}, the arguments used in the proof of Theorem 3.4 do not depend on the specific form of the original constants. This implies, as in the proof of Theorem 3.4 (bottom of p. 40 and top of p. 41) that
\begin{align*}
P(\FDP(c)>\gamma) &\le P(\bigcup_{i =1}^{N(|I|)} \{PV_{(i)} \le \beta_i (|I|)\}) \le |I| \cdot \sum_{i=1}^{N(|I|)} \frac{\beta_i(|I|)-\beta_{i-1}(|I|)}{i}
\end{align*}
where the last bound is obtained by lemma \ref{lemma:BasicLemma}.
\end{proof}
\begin{coro}\label{coro:FDP.SD.bound.1}
Let $c \in\mathcal{C}$ and $\beta$ be defined as in proposition \ref{prop:FDP.SD.bound} and $\widetilde{A}$ as in definition \ref{def:MatrixFDPSD}. Denote by $\widetilde{A}_{i \cdot}$ the $i$-th row of $\widetilde{A}$ and for $I \subset \{1,\ldots,n\}$ define $\beta(|I|)=(\beta_1(|I|),\ldots,\beta_{\lfloor \gamma n \rfloor+1}(|I|),0, \ldots,0) \in \reell^n$. For $\alpha \in (0,1)$ it holds: If 
\begin{align*}
\widetilde{A}_{1\cdot}\cdot \beta(1)^t &\le \alpha\\
\widetilde{A}_{2\cdot}\cdot \beta(2)^t &\le \alpha\\
&\vdots \\
\widetilde{A}_{n\cdot}\cdot \beta(n)^t &\le \alpha
\end{align*}
then $\max_{I \subset \{1,\ldots,n\}} P(\{\FDP(c) > \gamma\})\le \alpha$.
\end{coro}
\begin{proof} For any set $I \subset \{1,\ldots,n\}$ of true hypotheses by proposition \ref{prop:FDP.SD.bound} the probability $P(\FDP(c) > \gamma)$ is bounded by
\begin{align*}
& |I| \cdot \left\{\frac{\beta_1(|I|)}{1} + \frac{\beta_2(|I|) -\beta_1(|I|)}{2} + \cdots + \frac{\beta_{N(|I|)}(|I|) -\beta_{N(|I|)-1}(|I|)}{N(|I|)} \right\} \notag\\
&= |I| \cdot \left\{ \beta_1(|I|)\cdot \left(1-\frac{1}{2} \right) + \cdots + \beta_{N(|I|)-1}(|I|)\cdot \left(\frac{1}{N(|I|)-1}-\frac{1}{N(|I|)}\right) + \frac{\beta_{N(|I|)}(|I|)}{N(|I|)} \right\}\\
&=\widetilde{A}_{|I|\cdot}\cdot \beta(|I|)^t.
\end{align*}
\end{proof}
\begin{coro} Let $\gamma \in [0,1)$.
\begin{itemize}
	\item[(a)] Let $1\le i\le n$. For any $\delta \in \reell^n$ and $\beta_m(i):=\delta_{k_i(m)}$ it holds that $\widetilde{A}_{i \cdot} \cdot \beta(i)^t=(A^\FDPSD \cdot \delta^t)_i$.
\item[(b)] For $\alpha \in (0,1)$ and $c \in\mathcal{C}$ it holds: If $||A^\FDPSD\cdot c||_{\infty} \le \alpha$ then $P(\FDP(c)> \gamma)\le \alpha$.
\end{itemize}
\end{coro}
\begin{proof} For (a) we have
\begin{align*}
\widetilde{A}_{i \cdot} \cdot \beta(i)^t &= \sum_{\ell=1}^{\lfloor \gamma n \rfloor+1} \widetilde{A}_{i\ell} \cdot \beta_\ell (i) = \sum_{\ell=1}^{\lfloor \gamma n \rfloor+1} \widetilde{A}_{i\ell} \cdot \delta_{k_i(\ell)} \qquad \text{(by definition of $\beta$)}\\
&=\sum_{j=1}^n  \delta_j \cdot \left( \sum_{\ell: k_i(\ell)=j} \widetilde{A}_{i\ell} \right) =\sum_{j=1}^n \delta_j \cdot \left( \sum_{\ell \in k_i^{-1}(\{j\})} \widetilde{A}_{i\ell} \right)\\
&=(A^{\textnormal{FDP-SD}}\cdot \delta^t)_i,
\end{align*}
where in the first equality of the second row the convention $\sum_{\varnothing}\widetilde{A}_{i\ell}=0$ was used.

For part (b) note that if $||A^\FDPSD\cdot c||_{\infty} \le \alpha$ then this means by part (a) that $\max(\widetilde{A}_{1\cdot}\cdot \beta(1)^t, \ldots , \widetilde{A}_{n\cdot}\cdot \beta(n)^t)\le \alpha$ for $\beta_m(i):=c_{k_i(m)}$ and the claim then follows from corollary \ref{coro:FDP.SD.bound.1}.
\end{proof}
Thus theorem \ref{theorem:MainTheorem}, statement (d) is proved since for $d \in \mathcal{F}(A^{\FDPSD}(\gamma))$ and $c=\alpha \cdot d$ it now follows $||A^\FDPSD\cdot c||_{\infty}= \alpha \cdot||A^\FDPSD\cdot d||_{\infty}\le \alpha$ and part (b) from the above corollary yields the result.
\subsection{Comments}
For the step-up $\kFWER$ and $\FDP$ procedures, \citet{RomanoShaikh2006AOS} have proved that the choice $D=||A \cdot c||_{\infty}$ (with associated matrix $A$) is the smallest possible constant one can use for rescaled procedures of the form $c/D$ and still maintain control of the corresponding error rates. The key ingredient to their proof is part (ii) of lemma \ref{lemma:BasicLemma}.

For $\kFWER$ step-down procedures \citet[Theorem 2.3 (ii)]{LehmannRomano2005} show that none of the Lehmann-Romano constants $c_i=\frac{k}{n+k-i}$ for $i>k$ can be improved without violating the $\kFWER$. For $\FDP$ step-down procedures, \citet{RomanoShaikh2006} give an example that suggests that $D=||A \cdot c||_{\infty}$ is very nearly the smallest possible constant $d$ such that $c/d$ still controls $\FDP$, but no proof is given that this constant possesses the same optimality property as in the step-up case.

The modified FDP procedures introduced in section \ref{sec:NewFDPProcedures} by construction can not be improved without violating the LR bounds, i.e. without leading to $||A \cdot \xi||_{\infty}>1$. However, it is unclear whether this can also imply $P(\FDP > \gamma) > \alpha$. In the step-up case, the arguments given by \citet{RomanoShaikh2006} depend crucially on considering only linear modifications of the original procedures. Therefore these arguments do not seem applicable to investigating whether the modified procedures from section \ref{sec:NewFDPProcedures} can be improved any further.

\section{Simulation study} \label{sec:SimulationStudy}
In this section we investigate the power of the different FDP procedures in a simulation study. We consider the following procedures:
\begin{itemize}
	\item FDP-BH-SU and its modified variant FDP-BH-SU (mod),
	\item FDP-RS-SU and its modified variant FDP-RS-SU (mod),
	\item FDP-BH-SD and its modified variant FDP-BH-SD (mod),
	\item FDP-RS-SD and its modified variant FDP-RS-SD (mod).
\end{itemize}
The goals of the study are three-fold:
\begin{enumerate}
	\item to compare the power of the modified procedures with their original counterparts,
	\item to compare the power between the modified procedures,
	\item to compare the best FDP procedure (if it exists) with FDR controlling procedures. 
\end{enumerate}
To make the last comparison more consistent, we use for the step-up direction the \citet{BenjaminiYekutieli01} procedure FDR-BY-SU with critical constants
\begin{align*}
c_i^{BY}&=c_i^{BY}(n)=c_i^{BH}/D, \quad \text{where} \quad D=1+\frac{1}{2}+ \cdots + \frac{1}{n}
\end{align*}
which controls the FDR under arbitrary dependence. For the step-down direction we use the rescaled BH constants obtained by \citet{GuoRao2008}, i.e. 
\begin{align*}
c_i^{GR}&=c_i^{GR}(n)=c_i^{BH}/D, \quad \text{where} \\
D& = \max_{i=1, \ldots,n} \frac{i}{n} \left\{\sum_{j=1}^{n-i+1} \frac{1}{j} + \frac{n-i}{n-i+1}- \frac{n-i}{n} \right\}.
\end{align*}
We denote this approach by FDR-GR-SD. Similarly to \citet{RomShaWoETh2008} we control the median FDP as an alternative to controlling the FDR. We do this at the $.05$-level, i.e. $P(\FDP>0.05)\le 0.5$, while the FDR procedures control the expectation $\erw(\FDP)\le 0.05$. As \citet{RomShaWoETh2008} point out, the median FDP is a less stringent measure than the FDR in the sense that the probability of the FDP exceeding $0.05$ can be much bigger when the median FDP is controlled than when the FDR is controlled. 

For MTPs there are several ways to measure power, see e.g. \citep[Section 1.2.10]{DudoitLaan2007}. We use average power, i.e. the average proportion of rejected false hypotheses, for comparing procedures. We assume equicorrelated multivariate normal test statistics, i.e. $T=(T_1, \ldots, T_n) \sim \NormVert(\mu, \Sigma)$ with $\mu_i=0 $ for $i=1, \ldots, |I|$, $\mu_i=d$ for $i= |I|+1, \ldots,n$ and $\Sigma_{ij}=1/2$ for $i \neq j$ and  $\Sigma_{ij}=1$ else. For the parameter $d$, three nonzero values were used: $d = 0.1, 1$ and $3$, reflecting small, moderate and large deviations from the null hypotheses. For each simulated vector of test-statistics $p$-values were calculated for the gaussian test of the null hypotheses $H^0_i: \mu_i =  0$ (two-sided).  The number of tests performed was set to one of the values $10$, $50$, $100$ and $500$ reflecting small, medium and (moderately) large multiplicity of tests. We used $20000$ simulations in the simulation study which gives a uniform upper bound for the standard errors of $0.0035$. 

Figure \ref{fig:FDPvsModifiedFDPAverageProportionRejections_20000} depicts the gains in average power of the modified FDP procedures over the original (rescaled) variants. 
\begin{sidewaysfigure}[htbp]
	\centering
		\includegraphics[width=1.00\textwidth]{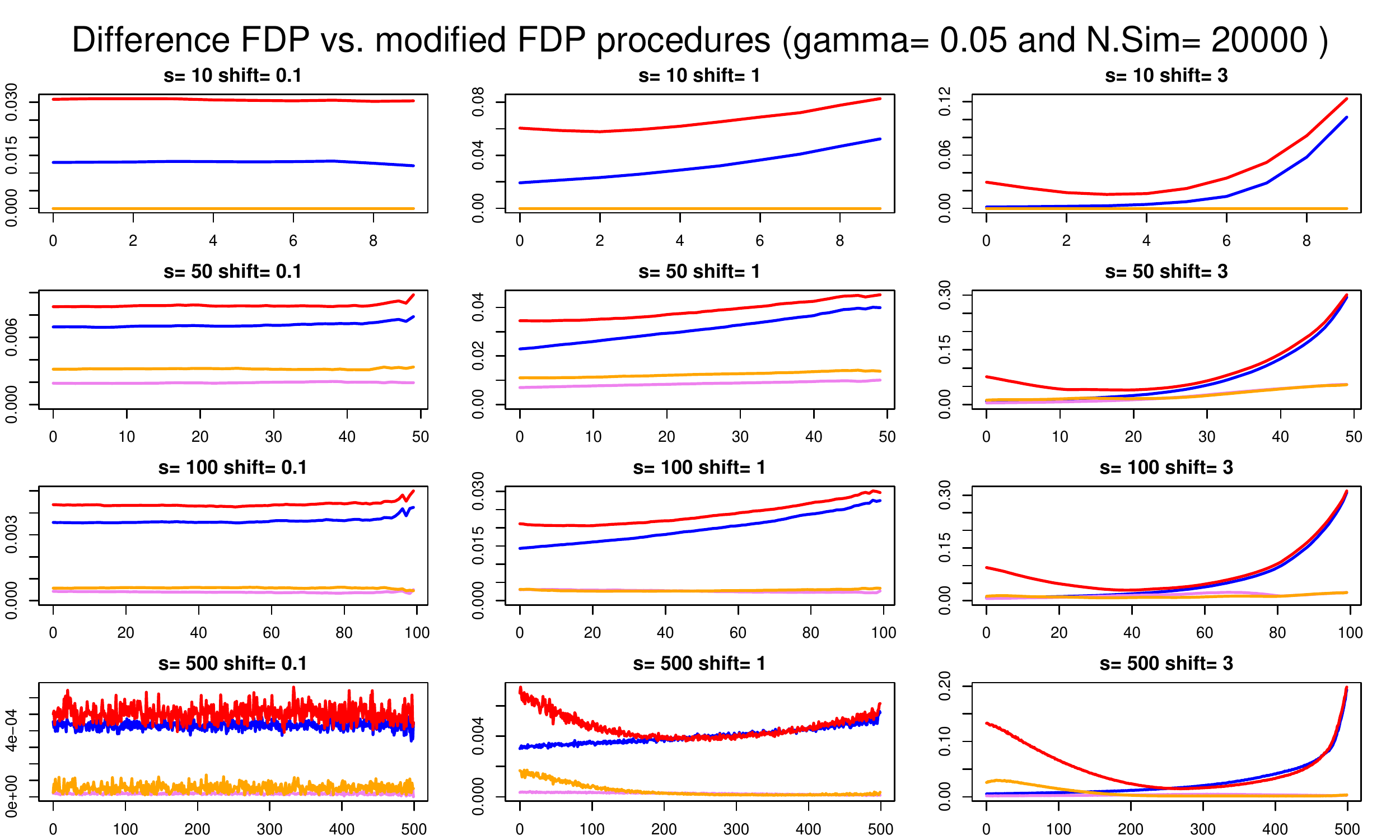}
	\caption{Difference of simulated average power for modified FDP procedures vs original (rescaled) FDP procedures (the $x$-axis is the number of true hypotheses). Shown are BH-SU (blue), RS-SU (violet), BH-SD (red), RS-SD (orange).}
	\label{fig:FDPvsModifiedFDPAverageProportionRejections_20000}
\end{sidewaysfigure} 
For most constellations, the gains in power are considerably larger for the BH-type procedures than for the RS-type procedures. Put differently, the RS procedures perform so well that in many situations none or only little improvement is possible.

\begin{sidewaysfigure}[htbp]
	\centering
		\includegraphics[width=1.00\textwidth]{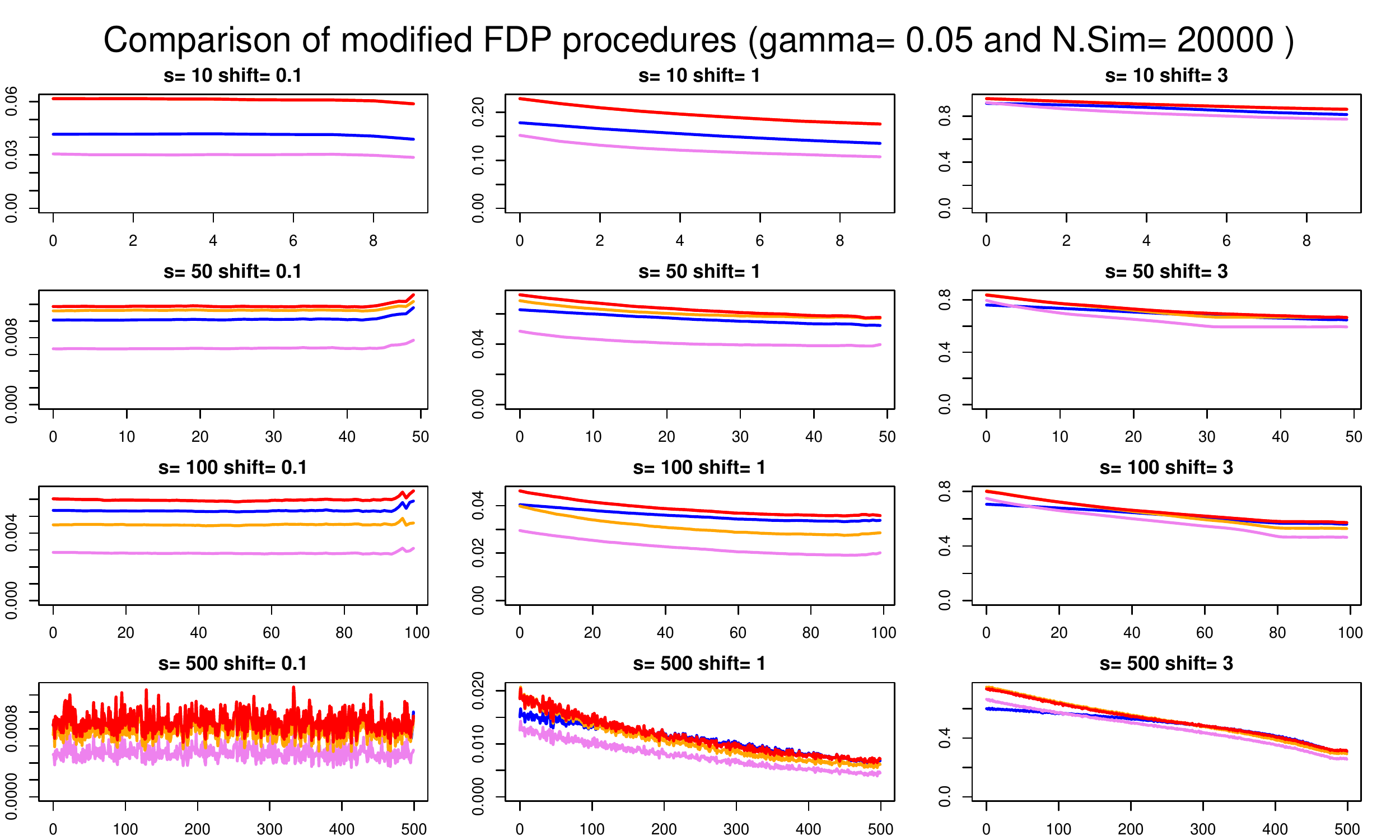}
	\caption{Simulated average power for modified FDP procedures. Shown are BH-SU (blue), RS-SU (violet), BH-SD (red), RS-SD (orange).}
	\label{fig:SelectedAverageProportionRejections_20000}
\end{sidewaysfigure} 

Figure \ref{fig:SelectedAverageProportionRejections_20000} presents a comparison of the four modified FDP-controlling procedures. The FDP-BH-SD procedure usually performs best and is followed closely by FDP-RS-SD and FDP-BH-SU. 
\begin{sidewaysfigure}[htbp]
	\centering
		\includegraphics[width=1.00\textwidth]{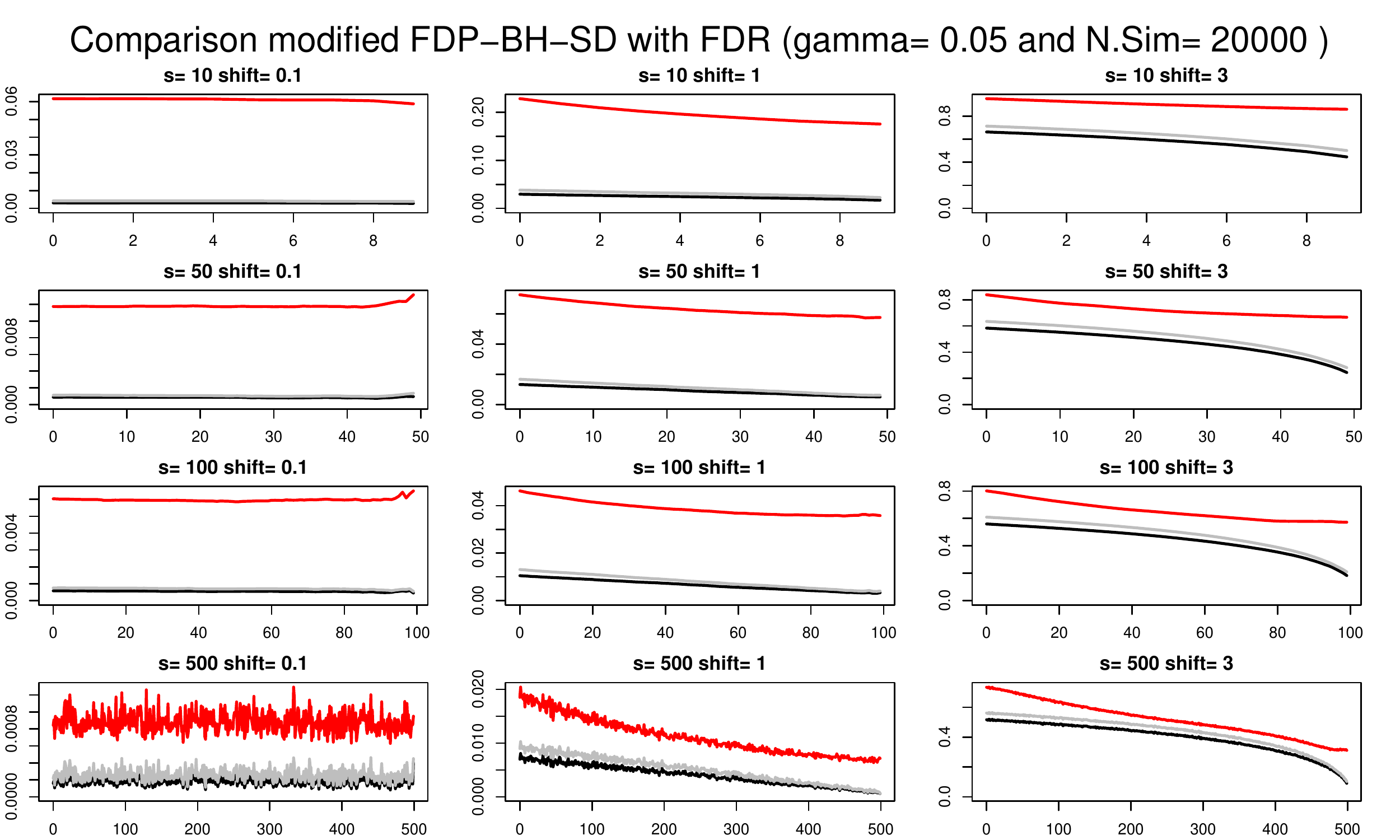}
	\caption{Simulated average power for modified FDP-BH-SD (red), FDR-BY-SU (black) and FDR-GR-SD (grey).}
	\label{fig:FDPvsFDRAverageProportionRejections_20000}
\end{sidewaysfigure} 

Figure \ref{fig:FDPvsFDRAverageProportionRejections_20000} compares the modified procedure FDP-BH-SD (mod) with the FDR procedures BY-SU and GR-SD. The median FDP-BH-SD posesses the highest power for all constellations while FDR-GR-SD and FDR-BY-SU perform very similarly.

Altogether we conclude that
\begin{itemize}
	\item modifying the rescaled FDP procedures resulted in increased power for all four procedures. The largest gains were achieved for the BH-type procedures,
	\item for the constellations considered here, FDP-BH-SD (mod) performed best, with FDP-SR-SD (mod) or FDP-BH-SU (mod) usually coming in a close second,
	\item the best modified median FDP procedure outperformed the FDR-controlling procedures that were rescaled in order to account for general dependence.
\end{itemize}
 
\section{Empirical applications}\label{sec:EmpiricalApplications}
In this section we compare the performance of the FDP and FDR approaches from the previous section for some empirical data.
\subsection{Benjamini-Hochberg data} \label{ssec:BHData}
We revisit the data analysed in \cite{BenjaminiHochberg95}, consisting of 15 $p$-values from a study on myocardial infarctation. Table \ref{tab:BHData} gives the results of applying the median FDP and FDR procedures at levels $q=0.05$ (note that in this case $\gamma-\FDP=\FWER$) and $q=0.10$, i.e. $P(\FDP>q)\le 0.5$ and $\erw(\FDP)\le q$.
\begin{table}[htb]
	\centering
\begin{tabular}{lcc}

           & \multicolumn{ 2}{c}{Number of rejections} \\ \cmidrule(lr){2-3}

    Method &     $q=0.05$ &     $q=0.10$ \\
\hline
\hline
 FDP-BH-SU &          9 &          9 \\

FDP-BH-SU (mod) &          9 &          9 \\

 FDP-RS-SU &          5 &          4 \\

FDP-RS-SU (mod) &          5 &          5 \\\hline

 FDR-BY-SU &          3 &         3 \\\hline

 FDP-BH-SD &         10 &         10 \\

FDP-BH-SD (mod) &         10 &         10 \\

 FDP-RS-SD &         10 &         10 \\

FDP-RS-SD (mod) &         10 &         10 \\\hline

 FDR-GR-SD &          3 &          4 \\
\hline
\end{tabular}    
\caption{Number of rejected hypotheses for the Benjamini-Hochberg data. }
\label{tab:BHData}
\end{table}
For $q=0.05$, the step-down procedures performed best, followed by the step-up FDP-BH and FDP-RS methods. The FDR procedures rejected the fewest hypotheses. Note that the FDP-RS-SU procedure rejects fewer hypotheses at level $0.10$ than at level $0.05$. This behavior is due to the fact that both the original constants and the scaling constant $D$ depend on the parameter $\gamma$. In this special case it means that $c_i^{0.05}\le c_i^{0.10}$ only for $i \in \{10,\ldots,14\}$. For the FDP-BH procedures this can not happen, since the original constants do not depend on the parameter $\gamma$ and the scaling constants are increasing in $\gamma$.  

\subsection{Westfall-Young data} \citet{WestYoung93} use resampling methods to analyze data from a complex epidemiological survey designed to assess the mental health of urban and rural individuals living in central North Carolina. The data consists of 72 raw $p$-values (see \citet[table 7.42]{WestYoung93}), with 25 of them $<0.05$ and 9 of the adjusted $p$-values $<0.05$. Table \ref{tab:WYData} displays the number of rejections when using the median FDP and FDR controlling procedures introduced above. 
\begin{table}[htb]
	\centering
\begin{tabular}{lcc}

           & \multicolumn{ 2}{c}{Number of rejections} \\

    Method &     $q=0.05$ &     $q=0.10$ \\
\hline
\hline
 FDP-BH-SU &          11 &         11  \\

FDP-BH-SU (mod) &          11 &          12 \\

 FDP-RS-SU &          10 &          11  \\

FDP-RS-SU (mod) &          11 &          11  \\\hline

 FDR-BY-SU &          10 &        10 \\\hline

 FDP-BH-SD &         11 &         11 \\

FDP-BH-SD (mod) &         12 &         12 \\

 FDP-RS-SD &         11 &         12 \\

FDP-RS-SD (mod) &         11 &         12 \\\hline

 FDR-GR-SD &          10 &          11 \\
\hline
\end{tabular}    
\caption{Number of rejected hypotheses for the Westfall-Young data. }
\label{tab:WYData}
\end{table}
All procedures reject at least one additional hypothesis. For level $q=0.05$, all median FDP procedures except RS-SU perform better than the FDR procedures; the modified BH-SD procedure is the only procedure that rejects three additional hypotheses. For $q=0.10$ the step-down FDP procedures seem to work best.

%
%
%
%
%
%
%
%
%
%
%

\subsection{Hedenfalk data} \label{ssec:HedenfalkData}
The data come from the breast cancer cDNA microarray experiment of Hedenfalk et al. (2001). In the original experiment, comparison was made between 3,226 genes of two mutation types, BRCA1 (7 arrays) and BRCA2 (8 arrays). The data included here are $p$-values obtained from a two- sample t-test analysis on a subset of 3,170 genes, as described in \citet{Storey2003}. Table \ref{tab:HedenfalkData} gives the results of applying the median FDP and FDR procedures at levels $q=0.05$ and $q=0.10$, i.e. $P(\FDP>q)\le 0.5$ and $\erw(FDP)\le q$.
\begin{table}[htb]
	\centering
\begin{tabular}{lcc}

           & \multicolumn{ 2}{c}{Number of rejections} \\

    Method &     $q=0.05$ &     $q=0.10$ \\
\hline
\hline
 FDP-BH-SU &          0 &          1\\

FDP-BH-SU (mod) &          6 &         10 \\

 FDP-RS-SU &          3 &          3 \\

FDP-RS-SU (mod) &          3 &          3\\\hline

 FDR-BY-SU &          0 &         1 \\\hline

 FDP-BH-SD &         0 &         1 \\

FDP-BH-SD (mod) &         7 &         4 \\

 FDP-RS-SD &         6 &         4 \\

FDP-RS-SD (mod) &         6 &         4 \\\hline

 FDR-GR-SD &          0 &          1 \\
\hline
\end{tabular}    
\caption{Number of rejected hypotheses for the Hedenfalk data. }
\label{tab:HedenfalkData}
\end{table}
\begin{figure}[htb]
    \subfigure{\includegraphics[width=0.49\textwidth]{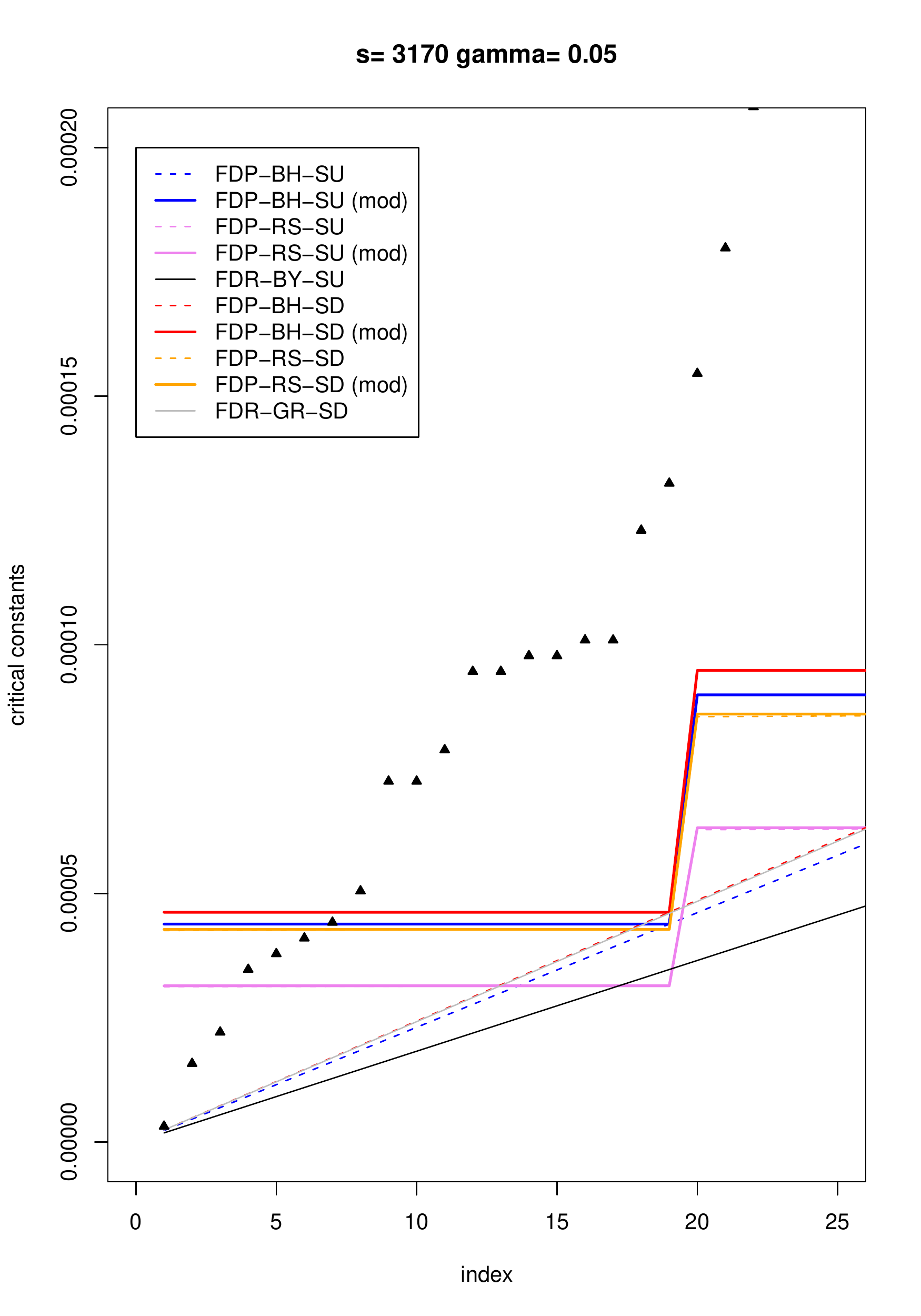}}
    \subfigure{\includegraphics[width=0.49\textwidth]{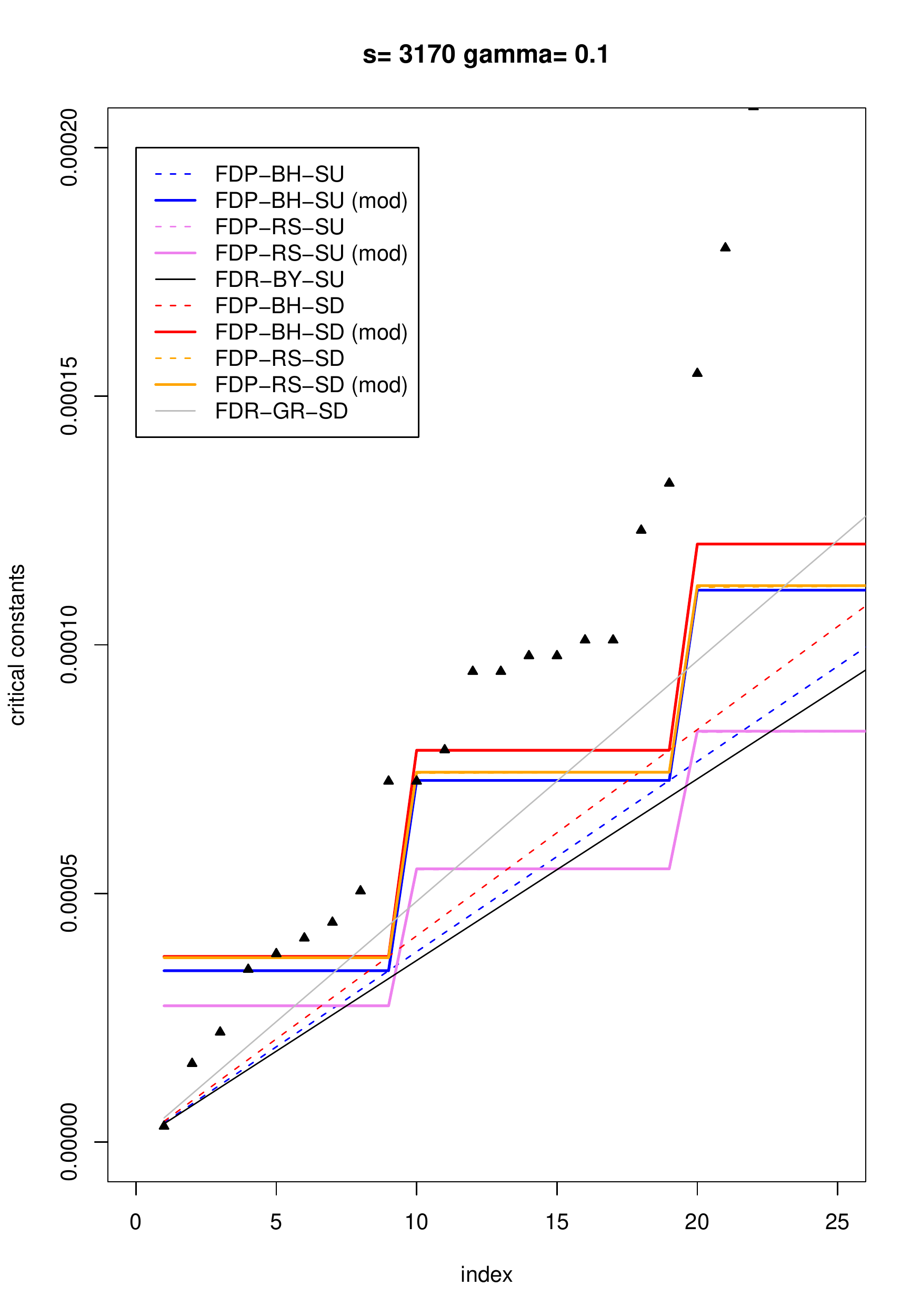}}
\caption{$p$-values (solid triangles) and multiple testing procedures for the Hedenfalk data (left panel: $\gamma=0.05$, right panel: $\gamma=0.01$).}
\end{figure}  
Again, the step-down FDP procedures perform better than their step-up counterparts, the modified median BH-SD procedure rejecting the most hypotheses. While all FDP procedrues except BU-SU reject more hypothese than both FDR approaches, we might hope for more powerful procedures. One alternative idea could be to use resampling methods in order to account for dependencies. However, as \citet{Pounds2006} points out, the power of these methods "will be severely limited, when the sample size is small". When the dependency between the $p$-values is assumed to be strong and extensive he tentatively recommends the FDR-BY-SU procedure.  

\section{Discussion} \label{sec:Discussion}
In this paper we have used results from \citet{RomanoShaikh2006,RomanoShaikh2006AOS} to obtain sufficient criteria for generalised error rates under general dependence in terms of systems of linear inequalities. These systems of linear inequalities describe the set of feasible points of a suitable linear optimisation problem. This property can be used to obtain modified multiple testing procedures which can improve on the rescaled procedures introduced in \citet{RomanoShaikh2006,RomanoShaikh2006AOS}. In a simulation study we have observed that these modified procedures can posess considerably more power than the original procedures.

While the focus of this work was on developing more powerful multiple testing procedures, \citet{HomBretz2008} have formulated additional desirable properties for such pocedures. Even though all methods considered here satisfy the property of coherence they are not particularly simple to describe and to communicate to non-statisticians. Since the modified procedures are obtained from a computationally complex numerical optimisation technique, the resulting sequence of critical constants will generally not exhibit any aesthetic mathematical patterns like e.g. the Bonferroni-Holm procedures. Another potential drawback from an aesthetical perspective may be related to what \citet{HomBretz2008} describe as  monotonicity properties of multiple testing procedures. While all procedures considered here yield monotonic decisions with respect to the corresponding type 1 error, it could be pointed out that additional monotonicity properties are conceivable that are not satisified by some of them. As a case in point, reconsider for $n=15$ the RS-procedures for $\gamma=0.05$ and $\gamma=0.10$ (see section \ref{ssec:BHData}). A numerical evaluation shows that $c^{0.05}_i<c^{0.10}_i$ holds true only for $i=10,\ldots,14$. Thus it may happen that more hypotheses are rejected for $\gamma=0.05$ than for $\gamma=0.1$ (at the same level of type 1 error) even though one would expect that the requirement $\FDP \le 0.05$ is more stringent than $\FDP \le 0.10$. The reason for this behavior is that both the original critical constants and the scaling constant $D$ depend on the parameter $\gamma$. For the FDP-BH procedures the original critical constants do not depend on $\gamma$. Numerical computations suggest that the scaling constants for FDP-BH are increasing in $\gamma$ and so this effect seems to be avoided by the FDP-BH procedures.

Another issue is the computational complexity of solving the linear programming problem which is needed to obtain the modified procedures. As a case in point, the calculation of the modified procedures used for analysing the Hedenfalk data with $n=3170$ in section \ref{ssec:HedenfalkData} took approximately nine hours on a Intel Xeon 5620 processor using the R-function \verb+Rglpk_solve_LP+. For multiple testing problems where the number of tests is significantly larger we thus expect run-time problems depending on the software and hardware available.

Finally, concerning other error rates like the FDR it seems natural to ask whether there are similar ways of modifying existing procedures under arbitrary dependence of the $p$-values. Recall that the key to modifying FDP procedures in section \ref{sec:NewFDPProcedures} was the observation that an improvement is possible whenever the $|I|^\star$-th row of matrix $A$ (with $|I|^\star= \arg\max_{|I|} (A \cdot c)_{|I|}$) contains at least one zero entry. \citet{GuoRao2008} have obtained bounds for the FDR that are similar to the bounds in theorem \ref{theorem:MainTheorem}, i.e. with $\FDR(c) \le \| A \cdot c\|_{\infty}$ for a suitable matrix $A$. However, as their Theorems 4.2 and 5.2 show, the corresponding step-up and step-down matrices do not contain any zero elements. Therefore, while the linear optimisation approach could still be used to define new procedures e.g. via an unconstrained linear program, it will not be possible to attain strict improvements along the lines of section \ref{sec:NewFDPProcedures}.

\bibliographystyle{chicago}
\label{bibliography}
\bibliography{Vorlage}

\end{document}